\documentclass[11pt]{article}
\usepackage{setspace}
\usepackage{amsmath,amssymb}
\usepackage{amsthm}
\usepackage{algorithm, algpseudocode}
\usepackage{url}
\usepackage{fullpage}
\usepackage{color}
\usepackage{graphicx}

\newtheorem{claim}{Claim}[section]

\newtheorem{theorem}{Theorem}
\newtheorem{proposition}[claim]{Proposition}
\newtheorem{corollary}[claim]{Corollary}
\newtheorem{definition}[claim]{Definition}

\theoremstyle{definition}
\newtheorem{remark}[claim]{Remark}

\def\NP{{\rm NP}}
\def\RP{{\rm RP}}
\def\P{{\rm P}}
\def\entro{s}
\def\di{{\partial i}}
\def\tZ{\widetilde{Z}}

\def\cH{{\cal H}}

\def\<{\langle}
\def\>{\rangle}
\def\prob{{\mathbb P}}

\def\naturals{{\mathbb N}}
\def\E{{\mathbb E}} %expectation
\def\Var{{\sf{Var}}}

\def\reals{\mathbb{R}}
\def\sT{{\sf T}}
\def\E{\mathbb{E}}

\def\cP{{\cal P}}

\def\id{{\rm I}}
\def\eps{\varepsilon}

\def\Cov{{\rm Cov}}
\def\ind{{\mathbb{I}}}
\def\proj{{\sf P}}
\def\atanh{{\rm atanh}}

\def\hF{\widehat{F}}

\def\sign{{\rm sign}}
\def\brho{{\boldsymbol \rho}}

\def\htheta{\hat{\theta}}
\def\bhtheta{\boldsymbol{\hat{\theta}}}
\def\btheta{{\boldsymbol{\theta}}}
\def\bJ{{\boldsymbol J}}
\def\bv{{\boldsymbol{v}}}
\def\bs{{\boldsymbol{s}}}
\def\bts{{\boldsymbol{\tilde{s}}}}
\def\ts{\tilde{s}}
\def\bx{{\boldsymbol{x}}}
\def\by{{\boldsymbol{y}}}
\def\ba{{\boldsymbol{a}}}
\def\bu{{\boldsymbol{u}}}
\def\bb{{\boldsymbol{b}}}
\def\bT{{\boldsymbol{T}}}
\def\bX{{\boldsymbol{X}}}
\def\btX{{\boldsymbol{\widetilde{X}} }}

\def\tbtau{{\boldsymbol{\tilde{\tau}}}}
\def\btau{{\boldsymbol{\tau}}}
\def\bsigma{{\boldsymbol{\sigma}}}
\def\bzero{{\boldsymbol{0}}}
\def\hZ{\widehat{Z}}
\def\creals{\overline{\mathbb R}}

\author{Andrea~Montanari\footnote{Department of Electrical
    Engineering and Department of Statistics, Stanford University}}
\title{Computational Implications of\\ Reducing Data to Sufficient Statistics}
\begin{document}
\maketitle

\begin{abstract}
Given a large dataset and an estimation task, 
it is common to pre-process the data by reducing them to a set of
sufficient statistics.
This step  is often regarded as straightforward 
and advantageous (in that it simplifies
statistical analysis).  I show that  --on the contrary--  reducing
data to sufficient statistics can change a computationally tractable estimation problem into an
intractable one. I discuss connections with recent work in theoretical
computer science, and implications for some  techniques to estimate
graphical models.
\end{abstract}

\section{Introduction}

The idea of sufficient statistics is a cornerstone of statistical
theory and statistical practice. Given a dataset, evaluating a
set of sufficient statistics yields a concise representation that can
be subsequently used to design (for instance) optimal statistical
estimation procedures.
To quote a widely adopted textbook \cite{Lehmann1998}:
\begin{quote}
`It often turns out that some part of the data carries no information
about the unknown distribution and that $\bX$ can therefore be replaced
by some statistic $\bT= \bT(\bX)$ without loss of information.'
\end{quote}
The main point of the present paper is the following.
While  optimal statistical estimation can be performed solely on the basis
of sufficient statistics, \emph{reduction to sufficient statistics can
  lead to an explosion in computational complexity.} This phenomenon
is so dramatic that --after reduction to sufficient statistics-- a tractable
estimation task can become entirely intractable.

To be concrete, we shall consider the problem of estimating the
parameters of an exponential 
family over $\bx\in \{0,1\}^p$. Namely, we consider the following
density with respect to the counting measure on $\{0,1\}^p$:
\begin{align}
p_{\btheta}(\bx) = \frac{1}{Z(\btheta)}\, h(\bx) \, e^{\<\btheta,\bx\>}\,. \label{eq:ExponentialFamily}
\end{align}
Here $h:\{0,1\}^p\to\reals_{>0}$ will be assumed strictly positive,
$\btheta\in\reals^p$ and $\<\ba,\bb\> = \sum_{i=1}^pa_ib_i$ is the
standard scalar product in $\reals^p$. The partition function
$Z(\btheta)$ is determined by the normalization condition, yielding
$Z(\btheta) = \sum_{\bx\in\{0,1\}^p}h(\bx) \, e^{\<\btheta,\bx\>}$.

We assume to be given $n$ i.i.d. samples $\bX_1,\dots,\bX_n\sim
p_{\btheta}$ and want to estimate $\btheta$. 
Introduced the notation $\bX^{(n)} = (\bX_1,\dots,\bX_n)$ for
the dataset of $n$ samples, taking values in $\{0,1\}^{p\times n}$. 
A consistent estimator is a function $\bhtheta:\{0,1\}^{p\times
  n}\to\reals^p$
such that, almost surely with respect to $\bX_1,\dots,\bX_n\sim
p_{\btheta}$, we have
\begin{align}
\lim_{n\to\infty} \bhtheta(\bX^{(n)}) = \btheta\, .
\end{align}
A vector of sufficient
statistics for the model $p_{\btheta}$ is clearly given by the empirical average
\begin{align}
\bT(\bX^{(n)}) = \frac{1}{n}\sum_{\ell=1}^n \bX_{\ell}\, .\label{eq:SufficientStatistics}
\end{align}
Classical statistical theory shows that, for any estimator $\bhtheta$
as above, we can construct a new estimator $\bhtheta^{\rm new}$ with
the \emph{same statistical properties}, that depends only on the
sufficient statistics $\bT(\bX^{(n)})$. In particular, 
$\bhtheta^{\rm new}:[0,1]^p\to\reals^p$ and, if  $\bhtheta$ is
consistent, then $\bhtheta^{\rm new}$ is also consistent. Namely,
$p_{\btheta}$-almost surely
\begin{align}
\lim_{n\to\infty} \bhtheta^{\rm new}(\bT(\bX^{(n)})) = \btheta\, .
\end{align}
In this paper we are concerned with the computational complexity of
evaluating such consistent estimators. In particular we will 
ask whether, for a given estimation task, there exist consistent
estimators that can be computed in time polynomial in the model
dimensions $p$. (We refer to Section \ref{sec:Main} for formal definitions.)

Our main result is that there exist simple exponential families of
the form  (\ref{eq:ExponentialFamily}) for which (under standard
computational complexity assumptions) no polynomial-time estimator
exists that uses only sufficient statistics. On the other hand, for
the same models, polynomial time  estimators exist that use
the whole data $\bX^{(n)}$.

In other words, the estimation problem is computationally tractable if
we use the whole dataset, but becomes intractable once the
data are reduced to sufficient statistics. This suggest that reducing
data to sufficient statistics might be harmful.

\subsection{Classical argument for sufficient statistics}

It is instructive to review the standard argument used to prove that a
reduction in complexity entails no loss of information
\cite{Lehmann1998}, since  this is
already suggestive of a possible explosion in computational complexity. Given any estimator
$\bhtheta(\bX^{(n)})$ of the parameters $\btheta$, the argument constructs a randomized estimator 
$\bhtheta^{{\rm new}}(\bT)$ that only depends on the sufficient statistics $\bT$,
and has the same distribution as $\bhtheta(\bX^{(n)})$, in particular the same
risk.
The construction is fairly simple. Given $\bT$, we can sample
$\btX^{(n)}$ (an independent copy of $\bX^{(n)}$)
from the law $p_{\btheta}(\,\cdot\,|\bT)$, conditional on the observed sufficient
statistics $\bT(\btX^{(n)})=\bT$. By definition of sufficient
statistics, the conditional law
$p_{\btheta}(\,\cdot\,|\bT)$ does not depend on $\btheta$, and hence
we can  sample  $\btX^{(n)}$ without knowing $\btheta$. We then define
$\bhtheta^{{\rm new}}(\bT)\equiv \bhtheta(\btX^{(n)})$. This has
clearly the same risk as $\bhtheta(\bX^{(n)})$. In words, we were able to
generate new data as informative as the original one using the
sufficient statistics.

The problem with this argument is that sampling from the
conditional distribution of $p_{\btheta}$ given $\bT(\btX^{(n)})=\bT$
can be computationally hard. Indeed simple examples can be given for the
weight $h(\,\cdot\,)$  (see
Section \ref{sec:Example}) that make sampling from $p_{\btheta}(\,\cdot\,|\bT)$ --even approximately--
impossible unless $\RP=\NP$ ($\RP\neq\NP$ this is a standard
computational complexity
assumption that is analogous to $\P\neq\NP$, and better suited for
sampling problems, see Section
\ref{sec:Example} for more details).
In other words, this argument is based on a  reduction that is not
computationally efficient. As we shall see, in
some examples, no reduction is computationally efficient, not only
the particular one given here. 

\subsection{Some examples}
\label{sec:SomeExamples}

In order to build intuition on our main result, we present 
two examples of the model (\ref{eq:ExponentialFamily}). Both fits in
the general context of estimating discrete graphical models.
Ths problem has been the object of considerable interest across
statistics and machine learning during recent years 
\cite{bresler2008reconstruction,bento2009graphical,ravikumar2010high,anandkumar2012high}.
Our intractability results imply that, for both of these examples, one
can construct families of parameters that cannot be estimated in
polynomial time under standard computational assumptions.

\vspace{0.4cm}

\noindent{\bf Ising model (unknown magnetic fields).} Let $\bJ = (J_{ij})_{i,j\in [p]}$ be a
symmetric matrix $\bJ=\bJ^{\sT}$. 
For $x\in\{0,1\}$, let $s(x) = 1-2x\in\{+1,-1\}$.
The associated Ising model is the
probability measure on $\{0,1\}^p$ with density function:
\begin{align}
p_{\bJ,\btheta}(\bx) = \frac{1}{\widetilde{Z}(\bJ,\btheta)}\,
\exp\Big\{\sum_{1\le  i<j\le p}
J_{ij}s(x_i)s(x_j) - \frac{1}{2}\sum_{i=1}^p\theta_is(x_i)\Big\}\, . \label{eq:IsingModelDef}
\end{align}
Ising models are one of the best studied models in statistical
physics\cite{ising1925beitrag,grimmett2006random}.
Since the seminal work of Hopfield \cite{hopfield1982neural} and Hinton
and Sejnowski \cite{ackley1985learning}, they found application in numerous
 areas of machine learning, computer vision, clustering and spatial statistics.  

Following the statistical physics terminology, the pairwise parameters
$J_{ij}$ are called `couplings', and the singleton parameters
$\theta_i$ are referred to as `magnetic fields.' (The standard
normalization for the latter does not include the factor $(-1/2)$,
that is introduced here for convenience.) The probability distribution
(\ref{eq:IsingModelDef})
can be written in the form (\ref{eq:ExponentialFamily}) by letting
\begin{align}
h(\bx) = \exp\Big\{\sum_{1\le i<j\le p}
J_{ij}s(x_i)s(x_j) \Big\}\, .\label{eq:HIsing}
\end{align}
In other words, the problem of estimating the parameter $\btheta$ can
be described as estimating the magnetic fields in an Ising model with
\emph{known couplings}.

Special choices of the couplings $\bJ$ have attracted most of the work
within statistical physics and theoretical computer science. A
possible construction is based on a graph $G = (V,E)$ with vertex set
$V=[p]=\{1,\dots,p\}$. We then set
\begin{align}
J_{ij} = \begin{cases}
\beta & \mbox{ if $(i,j)\in E$,}\\
0 & \mbox{otherwise.}
\end{cases}
\end{align} 
For $\beta>0$ this is the so-called \emph{ferromagnetic Ising
  model}. For $\beta<0$, we get the  \emph{anti-ferromagnetic Ising
  model}. In Section \ref{sec:AntiFerro} we will show that estimating 
the parameters $\btheta$ from sufficient statistics in the latter
model is --in general-- computationally intractable. As a
consequence, the problem is
intractable for the more general model (\ref{eq:HIsing}).

\vspace{0.4cm}

\noindent{\bf Ising model (unknown couplings).} 
A natural question --in view of the last example-- is whether our
frameworks sheds any light into the hardness of estimating the couplings
of an Ising model. 
In order to clarify this point, consider the case
of a zero-field model on variables $\bts =
(\ts_0,\ts_1,\ts_2,\dots,\ts_p)\in\{+1,-1\}^{p+1}$,
with coupling matrix $\bJ = (J_{ij})_{i,j\in \{0,1,\dots,p\}}$ and no
magnetic field:
\begin{align}
q_{\bJ}(\bts) = \frac{1}{Z(\bJ)}\, \exp\Big\{\sum_{0\le i<j\le
  p}J_{ij}\ts_i\ts_j \Big\}\, .
\end{align}
Note that $-\bts$ is distributed as $\bts$. Hence, without loss of
generality, we can preprocess the data flipping them
arbitrarily. Equivalently, we can consider the distribution of the vector
$\bs = (s_1,s_2,\dots, s_p)$, whereby $s_i = \ts_0\ts_i$. If $\bts$
has distribution $q_{\bJ}$, then $\bs$ has distribution
\begin{align}
p_{\bJ}(\bs) = \frac{1}{Z(\bJ)}\, \exp\Big\{\sum_{1\le i<j\le
  p}J_{ij}s_is_j +\sum_{i=1}^p J_{0i}s_i\Big\}\, .
\end{align}
Let us now assume that all the couplings have been estimated except
the $(J_{0i})_{i\in [p]}$. Accordingly, we define the parameters
$\theta_i = -2J_{0,i}$ and rewrite the above distribution as
\begin{align}
p_{\bJ_{\setminus 0},\btheta}(\bs) = \frac{1}{Z(\bJ_{\setminus 0},\btheta)}\, \exp\Big\{\sum_{1\le i<j\le
  p}J_{ij}s_is_j -\frac{1}{2}\sum_{i=1}^p \theta_is_i\Big\}\, ,
\end{align}
where $\bJ_{\setminus 0}\equiv (J_{ij})_{i,j\in [p]}$.
This coincides with Eq.~(\ref{eq:IsingModelDef}) apart from the
trivial change of variables $s_i = s(x_i)$. It follows that estimating
a row of the coupling matrix, with knowledge of the other couplings,
in absence of singleton terms  (magnetic fields)  is equivalent to estimating the
singleton terms knowing all the couplings (in a
model with one less variable).

The construction in Section \ref{sec:AntiFerro} shows that there are
choices of the couplings $\bJ_{\setminus 0}$ for which estimating the
parameters $\btheta$ from sufficient statistics is
intractable. Equivalently,  it is (in general)  intractable to
estimate $(J_{0i})_{i\in [p]}$ using the knowledge of
$(J_{ij})_{i,j\in [p]}$, and of the empirical averages of $(s_i)_{i\in [p]} =(\ts_0\ts_i)_{i\in]p]}$.

\subsection{Paper outline}

The rest of the paper is organized as follows. 
In Section \ref{sec:Main} we state formally our main results. Under
technical assumptions this shows that, if there is a computationally efficient
estimator $\bhtheta(\bT)$ that use only sufficient statistics, then
the partition function $Z(\bzero)$ can be approximated, also
efficiently. In Section \ref{sec:Example} we construct a family of
functions $h(\,\cdot\,)$ for which approximating  $Z(\bzero)$ in polynomial
time is impossible unless $\RP=\NP$.  As a consequence of our main
theorem, no efficient parameter estimator from sufficient statistics
exists for these models (unless $\RP=\NP$). Remarkably, we show that a
simple consistent estimator can be developed using the data
$\bX^{(n)}$, for the same models.
 
The example in Section \ref{sec:Example}  is an undirected graphical model,
and indeed intractability of computing approximations to $Z(\bzero)$
is quite generic in this context. 
In Section \ref{sec:Discussion}  we discuss relations to the
literature in algorithms and graphical models estimation.

\subsection{Notations}

We will use $[n] = \{1,2,\dots, n\}$  to denote the set of first $n$ integers.
Whenever possible, we will follow the convention of denoting
deterministic values by lowercase letters (e.g. $x,y,z,\dots$) and
random variables by uppercase letters (e.g. $X,Y,Z,\dots$). We reserve
boldface for vectors and matrices, e.g. $\bx$ is a deterministic
vector or matrix, and $\bX$ a random vector or matrix. 
For a vector $\bx\in\reals^m$, and a set $A\subseteq [m]$,
$\bx_A$ denotes the subvector indexed by $A$.
The components of $\bx$ are denoted by $(x_1,x_2,\dots,x_m)$.

Given  $f:\reals^m\to\reals$, we denote by $\nabla f(\bx)\in\reals^m$
its gradient at $\bx$, and by $\nabla^2 f(\bx)\in\reals^{m\times m}$
it Hessian. Whenever useful, we add as a subscript the variable with
respect to which we are differentiating as, for instance, in
$\nabla_{\bx}f(\bx)$. 
%
%**************************************************
%
\section{Computational hardness of consistent parameter estimation}
\label{sec:Main}

As mentioned above, we consider the problem of consistent estimation
of the parameters $\btheta$ of the model (\ref{eq:ExponentialFamily}),
from the sufficient statistics
(\ref{eq:SufficientStatistics}). Throughout we will assume $h(\bx)>0$
strictly for all $\bx\in\{0,1\}^p$.

As $n\to\infty$, $\bT(\bX^{(n)})$ converges --by the law of large
numbers-- to the population mean
\begin{align}
\btau_*(\btheta) \equiv \E_{\btheta}\{\bX\} =
\frac{1}{Z(\btheta)}\sum_{\bx\in\{0,1\}^p}\bx\, h(\bx) \,
e^{\<\btheta,\bx\>}\, .\label{eq:TauDef}
\end{align}
Here and below $\E_{\btheta}$ denotes expectation with respect to the
distribution $p_{\btheta}$.
We use the $*$ subscript to emphasize that this $\btau_*$ is a
function, and distinguish it from a specific value $\btau\in (0,1)^p$.
A consistent estimator based on the sufficient statistics $\bT$ must
necessarily  invert this function in the limit $n\to\infty$. This
motivates the following definition.
\begin{definition}\label{def:Estimator}
For $p\in\naturals$, let $\cH_p = \{h(\,\cdot\,)\}$ be a set of functions $h:\{0,1\}^{p}\to\reals_{>0}$.
A \emph{polynomial consistent estimator from sufficient statistics} for the model
(\ref{eq:ExponentialFamily}) with $h\in\cH_p$
is an algorithm that given $\btau= (\tau_1,\tau_2,\dots,\tau_p)\in (0,1)^p$, and a precision
parameter $\xi>0$, returns $\bhtheta(\btau,\xi)$ such that 
\begin{enumerate}
\item $\|\bhtheta(\btau,\xi)-\btheta\|_2\le \xi$ for some $\btheta$
  such that  $\btau_*(\btheta) = \btau$.
\item The algorithm terminates in time polynomial in $1/\xi$,
  $p$, the maximum description length of any of the $\tau_i$, and the
  description length of $h\in\cH_p$.
\end{enumerate}
\end{definition}
Two remaks are in order. The terminology adopted here (in particular,
the use of `consistent' and `polynomial') is motivated by the
following remarks. First, in the terminology of complexity theory,
this is a `fully polynomial time approximation scheme.' 

Second,  as discussed below, there is indeed a unique, continuous,
function $\btheta_*:(0,1)^p\to\reals^p$ such that
$\btau_*(\btheta_*(\btau)) =\btau$. This implies that $\btheta$ is
indeed a consistent estimator in the sense that for a sequence
$\xi_n\to 0$, we have, almost surely and
hence in probability,
\begin{align}
\lim_{n\to\infty} \bhtheta(\bT(\bX^{(n)}),\xi_n) = \btheta\, .
\end{align}
Finally, in the following we shall occasionally drop the qualification
`from sufficient statistics' when it is clear that we are considering
estimators that only use sufficient statistics.

We next state our main theorem.
\begin{theorem}\label{thm:Main}
Assume $\cH_p$, $p\ge 1$,  to be a family of weight functions 
 such that the the following three conditions hold:
\begin{enumerate}
\item[{\sf C1.}] There exists $\delta\in (0,1/2)$ such that, for any
  $p\in\naturals$, any $h\in\cH_p$, and any $i\in [p]$
\begin{align}
\tau_{*,i}(\btheta=\bzero)\in (\delta, 1-\delta)\, .
\end{align}
\item[{\sf C2.}]There exists a polynomial $L(p)$ such that, for any
  $p\in\naturals$, any $h\in\cH_p$, we have, for all $\btheta$ such
  that $\btau_*(\btheta)\in [\delta,1-\delta]^p$,
\begin{align}
\Cov_{\btheta}(\bX)\equiv \E_{\btheta}\big\{\bX\bX^{\sT}\big\}-
\E_{\btheta}\big\{\bX\big\}
\E_{\btheta}\big\{\bX\big\}^{\sT}\succeq\frac{1}{L(p)}\, \id_{p\times p}\, .\label{eq:ConditionC2}
\end{align}
\item[{\sf C3.}] There exists a polynomial $K(p)$ such that, for any
  $p\in\naturals$, any $h\in\cH_p$
\begin{align}
\max_{\bx,\by\in \{0,1\}^p}\big|\log h(\bx)-\log h(\by)\big|\le K(p)\, .\label{eq:ConditionC3}
\end{align}
\end{enumerate}
If there exists a polynomial consistent estimator for the model
$\cH=(\cH_p)_{p\ge 1}$,
then, for any $\eps>4p\delta ( K(p)+\log(4/\delta))$,  there exists an algorithm polynomial in the description length
of $h$, and in $(1/\eps)$, returning $\hZ$ such that 
\begin{align}
e^{-\eps}\, Z(\bzero)\le \hZ\le e^{\eps}\, Z(\bzero)\, . \label{eq:ApproxBound}
\end{align}
\end{theorem}

The main implication of this theorem is negative. 
If the problem of approximating $Z(\bzero)$ is intractable, it follows
from the above that there is no polynomial consistent estimator.
We will show in Section \ref{sec:Example} that we can use
recent results in complexity theory to construct a fairly simple class $\cH = \{h\}$
such that:
\begin{itemize}
\item An approximation scheme does not exist for $Z(\bzero)$
under a standard complexity theory  assumption known as $\RP\neq \NP$
(and analogous to $\P\neq \NP$, see Section \ref{sec:Example}). 
\item Remarkably, a consistent and computationally efficient estimator
  $\bhtheta = \bhtheta(\bX^{(n)})$ of the model parameters exists!
  However, this is not a function only of  the sufficient statistics. 
\end{itemize}
As a direct consequence, we have the following.
\begin{corollary}\label{coro:Intractable}
There are classes of models $\cH$ for which no polynomial consistent
estimator of the parameters from sufficient statistics exists, unless $\RP=\NP$.
\end{corollary}

\begin{remark}
Note that, neither the assumption nor the conclusion of Theorem
\ref{thm:Main} (or of Corollary \ref{coro:Intractable}) depend on the
number of samples $n$. 
In particular, the  map $\btau\mapsto\bhtheta(\btau,\xi)$ in Definition
\ref{def:Estimator} is a map  that approximately inverts  the moment map
$\btheta\mapsto \btau_*(\btheta) = \E_{\btheta}\{\bX\}$. 

Therefore, Theorem
\ref{thm:Main} and Corollary \ref{coro:Intractable} can be read as
stating that the map $\btheta\mapsto \btau_*(\btheta)$ cannot be
(approximately) inverted in polynomial time.
\end{remark}

The rest of this section is devoted to the proof of Theorem \ref{thm:Main}.
It is useful to first recall a few well-known properties of  exponential
families, as specialized to the present setting. While these are
consequences of fairly standard general statements (see for instance
\cite{efron1978geometry,Lehmann1998}), 
we present self-contained proofs in Appendix  \ref{app:Expo} for the readers'
convenience.

We will use standard statistics notations for the \emph{cumulant generating function}
(also known as  \emph{log-partition function}), $A:\reals^p\to\reals$,
\begin{align}
A(\btheta) \equiv \log Z(\btheta)  = \log\Big\{\sum_{\bx\in\{0,1\}^p}
h(\bx) \, e^{\<\btheta,\bx\>}\Big\}\, ,
\end{align}
and its Legendre-Fenchel transform $F:(0,1)^p\to\reals$, which we
shall call the \emph{free energy},
\begin{align}
F(\btau) = \inf_{\btheta\in\reals^p}
\big(A(\btheta)-\<\btau,\btheta\>\big)\, .
\end{align}
\begin{proposition}\label{prop:Expo}
Given a strictly positive $h:\{0,1\}^p\to\reals_{>0}$, the following
hold:
\begin{enumerate}
\item[{\sf Fact1.}] The function $\btau_*: \reals^p\to (0,1)^p$ is $C_{\infty}(\reals^p)$.
\item[{\sf Fact2.}] Recalling that
  $\Cov_{\btheta}(\bX)\in\reals^{p\times p}$ denotes the covariance
  matrix of the law $p_{\btheta}$, we have
\begin{align}
\nabla_\btheta A(\btheta) & = \btau_*(\btheta)\, ,\\
\nabla_{\btheta}^2A(\btheta) & = \nabla_{\btheta}\btau_*(\btheta)  = \Cov_{\btheta}(\bX)\, .\label{eq:HessianCov}
\end{align}
\item[{\sf Fact3.}]$\Cov_{\btheta}(\bX)\succeq
  c(\btheta)\, \id_p$ for some continuous strictly positive $c(\btheta)>0$.
\item[{\sf Fact4.}] $\btau_*: \reals^p\to (0,1)^p$ is a bijection. We will denote by $\btheta_*:(0,1)^p\to\reals^p$ the
  inverse mapping.
\item[{\sf Fact5.}] The function $F:(0,1)^p\to\reals^p$ is concave and
  $C_{\infty}((0,1)^p)$ with 
\begin{align}
A(\btheta) & = \max_{\btau \in
  (0,1)^p}\big\{F(\btau)+\<\btau,\btheta\>\big\} = F(\btau_*(\btheta))+\<\btau_*(\btheta),\btheta\>\, ,\label{eq:InverseTransform}\\
\nabla_{\btau} F(\btau) &= -\btheta_*(\btau)\,\\
\nabla_{\btau}^2 F(\btau)  & = - \nabla_{\btau}\btheta_*(\btau)  =
-\big[\nabla^2_{\btheta}A(\btheta_*(\btau))\big]^{-1}\, . \label{eq:InverseHessian}
\end{align}
\item[{\sf Fact6.}] Let $\cP_p = \cP(\{0,1\}^p)$ be the set of
  probability distributions over $\{0,1\}^p$, and denote by $H(q) =
  -\sum_{\bx\in\{0,1\}^p}q(\bx)\log q(\bx)$ the Shannon entropy of
  $q\in\cP_p$.
Then 
\begin{align}
F(\btau)  = \max_{q\in\cP_p}\Big\{H(q)+ \E_q\log h(\bX) \;\mbox{ such that }
\E_q\{\bX\} = \btau\Big\}\, . \label{eq:RestrictedGibbs}
\end{align}
\end{enumerate}
\end{proposition}

We are now in position to prove  Theorem \ref{thm:Main}: we present
here a version of the proof that omits some technical step, and
complete
the details in Appendix \ref{app:Details}.
\begin{proof}[Proof of Theorem \ref{thm:Main}.]
By Eq.~(\ref{eq:InverseTransform}) we have, letting $D_p(\delta)\equiv [\delta,1-\delta]^p$,
\begin{align}
\log Z(\bzero) &= \max_{\btau\in (0,1)^p} F(\btau) \\
&= \max_{\btau\in D_p(\delta)} F(\btau)  = F(\btau_*(\bzero))\, ,
\end{align}
where the second equality follows from assumption {\sf C1}.

The problem of computing $Z(\bzero)$ is then reduced to the one of
maximizing the concave function $F(\btau)$ over the convex set
$D_p(\delta)$. Note that, by assumption {\sf C2} and
Eq.~(\ref{eq:InverseHessian}), the gradient of $F(\btau)$ has Lipchitz
modulus bounded by $L(p)$ on $D_p(\delta)$. Namely, for all
$\btau_1,\btau_2\in D_p(\delta)$, 
\begin{align}
\big\|\nabla_{\btau}F(\btau_1)-\nabla_{\btau}F(\btau_2)\big\|_2 \le
L(p) \, \|\btau_1-\btau_2\|_2\, .
\end{align}

We will maximize $F(\btau)$ by a standard projected gradient algorithm. We
will work here under the assumption that we have access to an oracle
that given a point $\btau\in D_p(\delta)$, returns $\btheta_*(\btau) =
-\nabla_{\btau}F(\btau)$.
While in reality we do not have access to such an oracle, Definition
\ref{def:Estimator} (and the Theorems assumptions) imply that there
exists an efficient approximation scheme for this oracle. 
We will show in Appendix \ref{app:Details} that indeed we can replace
the oracle by such an approximation scheme.

Given the oracle, the projected gradient algorithm is defined by the
iteration
(with the superscript $t$ indicating the iteration number, and letting
$L=L(p)$)
\begin{align}
\btau^0 & = \Big(\frac{1}{2},\frac{1}{2},\dots
,\frac{1}{2}\Big)^{\sT}\, ,\\
\btau^{t+1} & =
\proj_{\delta}\Big(\btau^t-\frac{1}{L}\btheta_*(\btau^t)\Big)\, . \label{eq:ProjGrad}
\end{align}
Here $\proj_{\delta}(\bx)$ is the  orthogonal projector on
$D_p(\delta)$, which can be computed efficiently since, for each
$i\in\{1,2,\dots,p\}$,
and $\bu\in\reals^p$,
\begin{align}
\proj_{\delta}(\bu)_i = 
\begin{cases}
\delta & \mbox{ if } u_i<\delta\, ,\\
u_i & \mbox{ if } u_i\in [\delta,1-\delta]\, ,\\
1-\delta & \mbox{ if } u_i>1-\delta\, .
\end{cases}
\end{align}
We run $t_0 = t_0(p,\eps) \equiv \lceil 2p\,L(p)/\eps\rceil$ iterations of projected gradient. By
\cite[Theorem 3.1]{beck2009fast} we have
\begin{align}
0\le F(\btau_*(\bzero))-F(\btau^{t_0})\le
\frac{L\|\btau^0-\btau_*(\bzero)\|_2^2}{2t_0}\le \frac{\eps}{4}\, .\label{eq:ApproxMax}
\end{align}
Hence, $F(\btau^{t_0})$ provides  a good approximation of $F(\btau_*(\bzero)) = \log Z(\bzero)$. 

We are left with the task of evaluating $F(\btau^{t_0})$. The idea is
to `integrate' the derivative of $F(\btau)$ along a path that starts
at a point $\btau^{(0)}$ where $F$ can be easily approximated.
We will use again $\nabla_{\btau} F(\btau) = -\btheta_*(\btau)$ and
assume that $\btheta_*$ is exactly computed by an oracle.
Again we will see in Appendix \ref{app:Details} that this oracle can
be replaced by the estimator in Definition \ref{def:Estimator}
 
Let $\btau^{(0)} \equiv (\delta,\dots,\delta)^{\sT}$.
Next consider Eq.~(\ref{eq:RestrictedGibbs}). For any $q\in\cP_p$ such
that  $\E_q(\bX)=\btau^{(0)}$, we have, letting
$\entro(x) \equiv -x\log x-(1-x)\log (1-x)$
\begin{align}
&0\le H(q) \le p \, \entro(\delta)\, ,\label{eq:EntroBounds}\\
&\Big|\E_p\log \frac{h(\bX)}{h(\bzero)}\Big|\le K(p) \,
\prob_q(\bX\neq\bzero)\le K(p)p\delta\, . \label{eq:EnergyBounds}
\end{align}
Where the second inequality  in Eq.~(\ref{eq:EntroBounds}) follows
since the joint entropy is not larger than the sum of the entropy of
the marginals. The first inequality in Eq.~(\ref{eq:EnergyBounds})
follows from assumption {\sf C3}, and the last inequality in Eq.~(\ref{eq:EnergyBounds})
is Markov's.
Using these bounds in Eq.~(\ref{eq:RestrictedGibbs}), we get
\begin{align}
\Big|F(\btau^{(0)})-\log h(\bzero)\Big|\le p\, \entro(\delta)+pK(p)\,
\delta\le \frac{\eps}{4}\, ,\label{eq:BoundStart}
\end{align}
where the last inequality follows from the assumption $\eps>4p\delta ( K(p)+\log(4/\delta))$.

For an integer $m$, we let $\btau^{(m)} = \btau^{t_0}$ be the output
of projected gradient, and, for $\ell\in\{1,\dots,m-1\}$,
$\btau^{(\ell)} \equiv \btau^{(0)} + \ell\,
(\btau^{(m)}-\btau^{(0)})/m$ be given by linearly interpolating
between $\btau^{(0)}$ and $\btau^{(m)}$. Note that,
since $\btau^{(0},\btau^{(m)}\in D_p(\delta)$, we have
$\btau^{(\ell)}\in D_p(\delta)$ as well, by convexity.

 We finally construct our
approximation $\hZ$ by letting 
\begin{align}
\log \hZ^{{\rm or}} \equiv \log h(\bzero) -
\sum_{\ell=1}^{m}\<\btheta_*(\btau^{(\ell)}),\btau^{(\ell)}-\btau^{(\ell-1)}\>\,. \label{eq:Zestimate}
\end{align}
(We introduced the superscript ${\rm or}$ to emphasize that this
approximation makes use of the oracle $\btheta_*$. In Appendix
\ref{app:Details} we control the additional error induced by the use
of $\bhtheta$.)
Let us bound the approximation error:
\begin{align}
\Big|\log \frac{\hZ^{{\rm or}}}{Z(\bzero)}\Big|
\le &\big|F(\btau^{(0)})-\log h(\bzero)\big|\label{eq:BoundSum}\\
&+ \sum_{\ell=1}^{m}\big|F(\btau^{(\ell)})-F(\btau^{(\ell-1)})-
\<\nabla_{\btau}F(\btau^{(\ell)}),\btau^{(\ell)}-\btau^{(\ell-1)}\>\big|\nonumber\\
&+\big|F(\btau_*(\bzero)) - F(\btau^{(m)}) \big|\, .\nonumber
\end{align}
The first term is bounded by Eq.~(\ref{eq:BoundStart}) and the last by 
Eq.~(\ref{eq:ApproxMax}). As for the middle term, by the intermediate
value theorem there exists, for each $\ell$, a point
$\tbtau^{(\ell)}\in [\btau^{(\ell-1)},\btau^{(\ell)}]$ such that
\begin{align}
\big|F(\btau^{(\ell)})-F(\btau^{(\ell-1)})-
\<\nabla_{\btau}F(\btau^{(\ell)}),\btau^{(\ell)}-\btau^{(\ell-1)}\>\big|
&= \big|
\<\nabla_{\btau}F(\tbtau^{(\ell)})-\nabla_{\btau}F(\btau^{(\ell)}) ,\btau^{(\ell)}-\btau^{(\ell-1)}\>\big|\nonumber\\
\le L(p) \, \|\btau^{(\ell)}-\btau^{(\ell-1)}\|_2^2&\le
\frac{L(p)}{m^2}\, \|\btau^{(m)}-\btau^{(0)}\|_2^2 \le
\frac{L(p)p}{m^2}\, .
\end{align}
Hence, choosing $m= m_0(p,\eps) \equiv \lceil 4L(p)p/\eps\rceil$, we
can bound the sum in Eq.~(\ref{eq:BoundSum}) by $\eps/4$. Hence the
approximation error is bounded by
\begin{align}
\Big|\log \frac{\hZ^{{\rm or}}}{Z(\bzero)}\Big| \le
\frac{\eps}{4}+\frac{\eps}{4}+\frac{\eps}{4}\le \frac{3\eps}{4}\, ,
\end{align}
which concludes our proof outline.
\end{proof}

\section{An example}
\label{sec:Example}

In this section we describe a simple class of functions 
$\cH_p$ for which it is impossible to estimate $Z(\bzero)$ in
polynomial time unless
$\RP=\NP$. $\RP\neq\NP$ is a standard computational assumption
analogous to $\P\neq\NP$. Informally, $\NP$ is the set of decision problems
whose solution can be checked in polynomial time. On the other hand,
$\RP$ is the set of problems that can be solved in polynomial time
using a randomized algorithm (roughly, an algorithm that has access to
an ideal random numbers generator). The assumption $\RP\neq\NP$ posits
that  there are problems that cannot be solved in polynomial time even
using randomization. As such, it is a slightly stronger than $\P\neq
\NP$ but nearly as widely believed to be true. We refer to standard
textbooks, e.g.  \cite{arora2009computational}, for further details.

Using our  Theorem \ref{thm:Main}, we will show that
consistent parameter estimation using sufficient statistics is
intractable for the models
in the class $\cH_p$.
We will then show that --for the same models-- it is quite easy to
estimate the parameters from i.i.d. samples $\bX_1,\dots,\bX_n\sim p_{\btheta}$.

\subsection{Intractability of estimation from sufficient statistics}
\label{sec:AntiFerro}

For $\beta>0$ a \emph{fixed known number}, and $G = (V=[p], E)$ a simple graph, we
let $h_{G,\beta}:\{0,1\}^p\to \reals_{>0}$ be defined by
\begin{align}
h_{G,\beta}(\bx)\equiv \exp\Big\{2\beta\sum_{(i,j)\in E} \ind(x_i\neq
x_j)\Big\}\, . \label{eq:AntiFerro}
\end{align}
This is also known as the \emph{anti-ferromagnetic Ising model} on
graph $G$. (The same model was already introduced in Section
\ref{sec:SomeExamples},
with slightly different normalizations.)

Fixing $k\in \naturals$, and $\beta\in\reals_{>0}$, we introduce the
class of functions 
\begin{align}
\cH_p(k,\beta) \equiv \big\{ h_{G,\beta} : \, G \mbox{ is a simple
  regular graph of
  degree $k$ over $p$ vertices }\big\}\, .
\end{align}
(Recall that a regular graph is a graph with the
same degree at all vertices. The set of regular graphs is non-empty as
soon as $pk$ is even, and $p\ge p_0(k)$.)
Intractability of approximating $Z(\bzero)$ was characterized in
\cite{sly2012computational,galanis2012inapproximability}. We restate the main result of \cite{sly2012computational}, adapting
it to the present setting\footnote{The present statement differs from
  the original one in \cite{sly2012computational} in one technical
  aspect. We state that $Z(\bzero)$ cannot be approximated
  within a ratio $e^{\eps_0p}$, while \cite{sly2012computational}
  state their result in terms of FPRAS. However, the authors of
  \cite{sly2012computational} also mention that their same proof
  yields impossibility to achieve the approximation
  ratio stated here (which is in fact easy to check).}.
\begin{theorem}[Sly, Sun, 2012]\label{thm:SlySun}
For any $k\ge 3$ and $\beta> \atanh(1/(k-1))$ there exists $\eps_0 =\eps_0(k,\beta)>0$ such
that the following holds. Unless $\RP=\NP$, there is no polynomial
algorithm taking as input $h\in\cH_{p}(k,\beta)$ and returning $\hZ$
such that 
\begin{align}
e^{-\eps_0p}Z(\bzero)\le \hZ\le e^{\eps_0p}Z(\bzero)
\end{align} 
where we recall that $Z(\bzero) \equiv\sum_{\bx\in\{0,1\}^p} h(\bx)$.
\end{theorem}
We note in passing that \cite{sinclair2014approximation} implies that
this result is tight. Namely, for $\beta<\atanh(1/(k-1))$ there exists
a fully polynomial approximation scheme for $Z(\bzero)$. 

We can now state (and prove) and more concrete form of Corollary
\ref{coro:Intractable}.
In Section \ref{sec:EstimationFromSamples} we will show that the model
parameters can be consistently estimated in polynomial time from 
i.i.d. samples $\bX_1,\dots,\bX_n\sim p_{\btheta}$ .
\begin{corollary}
Fix $k\ge 3$ and $\beta> \atanh(1/(k-1))$. Unless $\RP=\NP$, then there
exists no polynomial consistent estimator from sufficient statistics for the model
$\{\cH_p(k,\beta)\}_{p\ge p_0(k)}$.
\end{corollary}
\begin{proof}
The proof consists in checking that the assumptions {\sf C1}, {\sf
  C2}, {\sf C3} of Theorem \ref{thm:Main} apply. It then follows from
Theorem \ref{thm:SlySun} that either $\RP=\NP$ or there is no polynomial
consistent estimator from sufficient statistics. Throughout, we will
write $C$ or $c$ for generic strictly positive constants that can depend on $\beta$, $k$.

\vspace{0.25cm}

Let us start with condition  {\sf C1}. We fix a graph $G = (V=[p],E)$ on $p$
vertices. For a vertex $i\in [p]$, we let $\di = \{j\in [p]: (i,j)\in
E\}$ be the set of neighbors of $i$.
 Writing $\prob_{\bzero}$ for
$\prob_{\btheta=\bzero}$, we have $\tau_{*,i}(\bzero) =\prob_{\bzero}(X_i=1)$.

Note that $p_{\btheta}$ is a Markov Random Field with graph $G$. We
therefore have
\begin{align}
\min_{\bx_{\di}\in\{0,1\}^{\di}}\prob_{\bzero}(X_i=1|\bX_{\di} =\bx_{\di})\le
\prob_{\bzero}(X_i =1) \le
\max_{\bx_{\di}\in\{0,1\}^{\di}}\prob_{\bzero}(X_i=1|\bX_{\di}
=\bx_{\di})\, .\label{eq:ConditionalBound}
\end{align}
A simple direct calculation with Eq.~(\ref{eq:AntiFerro}) yields
\begin{align}
\prob_{\bzero}(X_i=1|\bX_{\di} =\bx_{\di}) = \frac{e^{2\beta\,
    n_0(\bx_{\di})}}{e^{2\beta\,
    n_0(\bx_{\di})} +e^{2\beta\,
    n_1(\bx_{\di})}}\, ,\label{eq:CondProb}
\end{align}
where $n_0(\bx_\di)$ and $n_0(\bx_\di)$  are the number of zeros
and ones in the vector $\bx_{\di}$. 

The right hand side of Eq.~(\ref{eq:CondProb}) is maximized when $n_0(\bx_\di)= k-n_1(\bx_\di) =
k$, and minimized when $n_0(\bx_\di)= k-n_1(\bx_\di) = 0$. We then
have
\begin{align}
\frac{1}{1+e^{2\beta k}} \le \prob_{\bzero}(X_i=1)\le \frac{e^{2\beta
    k}}{1+e^{2\beta k}} \, .
\end{align}
We conclude that there exists $c = c(k,\beta)>0$ such that condition
{\sf C1} is satisfied for all $\delta\in (0,c)$. We will select the
value of $\delta$  after checking condition {\sf C3}.

\vspace{0.25cm}

Consider condition {\sf C3}. From Eq.~(\ref{eq:AntiFerro}), it follows
that $h_{G,\beta}(\bx)\ge 1$ (because the argument in the exponent is
non-negative) and $h_{G,\beta}(\bx) \le \exp(2\beta|E|)$ with $|E|$
the number of edges (the upper bound is saturated if the graph is
bipartite).
Since $|E| = kp/2$, Eq.~(\ref{eq:ConditionC3}) follows with $K(p) =
\beta\, k p$. 

At this point we choose $\delta = 1/(10pK(p))$. Assuming, without loss
of generality, $K(p), p\ge 10$, this implies that we can take $\eps = 2$
in Eq.~(\ref{eq:ApproxBound}), which yields the desired contradiction
with Theorem \ref{thm:SlySun}, provided we can verify condition {\sf
  C2} with the stated value of $\delta$.

\vspace{0.25cm}
 
To conclude our proof, consider condition {\sf C2}. First we claim
that $\btau_*(\btheta)\in [\delta,1-\delta]^p$ implies
$\|\btheta\|_{\infty}\le C\, \log p$  for some constant $C =
C(k,\beta)$. Indeed, let us fix $i\in [p]$ and prove that $\theta_i\le
C\,\log p$ (the lower bound follows from an analogous argument).
Using again the fact that $p_{\btheta}$ is a Markov Random Field, we
have (since, by definition $\tau_{*,i}(\btheta)=\prob_{\btheta}(X_i = 1)$)
\begin{align}
1-\delta\ge \prob_{\btheta}(X_i = 1)\ge
\min_{\bx_{\di}\in\{0,1\}^{\di}}\prob_{\btheta}(X_i=1|\bX_{\di}
=\bx_{\di}) \, .
\end{align}
Proceeding as in the case of condition {\sf C1}, we see that the last
conditional probability is minimized when all the neighbors of $i$ are 
in state $0$ (i.e. $\bx_{\di} = \bzero$). This yields
\begin{align}
1-\delta\ge \frac{e^{\theta_i}}{e^{\theta_i}+e^{2\beta k}}, 
\end{align}
which is equivalent to $\theta_i\le 2\beta k+\log((1-\delta)/\delta)$.
Substituting $\delta = 1/(10pK(p))$ with $K(p)$ a polynomial yields the
desired bound (recall that $\beta$ and $k$ are constants).

Next, we claim that,  there exists a polynomial $L_0(p)$ such that, for
each $i\in[p]$, all $\bx_{\di}$ and all $\btheta$ with $|\theta_i|\le C\log p$, we have
\begin{align}
\Var_{\btheta}(X_i|\bX_{\di}=\bx_{\di})\ge \frac{1}{L_0(p)}\, .\label{eq:VarI}
\end{align}
The calculation is essentially the same as the one already carried out
above for $\prob_{\btheta}(X_i = 1|\bX_{\di} = \bx_{\di})$ and
therefore we omit it.

Finally, we want to prove Eq.~(\ref{eq:ConditionC2}) for some
polynomial $L(p)$. Equivalently, we need to prove that
$\Var_{\btheta}(\<\bv,\bX\>)\ge 1/L(p)$ for any vector $\bv\in\reals^p$,
$\|\bv\|_2=1$.
Fix one such vector $\bv$. Let $i(1)$ be the index of the component of
$\bv$ with largest magnitude (i.e. $|v_{i(1)}|\ge |v_j|$ for all
$j\in [p]\setminus\{i(1)\}$). Let $i(2)$ be the of the component of
$\bv$ with largest magnitude \emph{excluded $i(1)$ and $\di(1)$}  (i.e. $|v_{i(2)}|\ge |v_j|$ for all
$j\in [p]\setminus\{i(1),\di(1),i(2)\}$), and so on.
Namely, for each $\ell$, we let $i(\ell)$ be the index of the
component of $\bv$ with the largest magnitude \emph{excluded
  $i(1),\dots,i(\ell-1)$ and their neighbors in $G$}.
Denote by $m\ge n/(k+1)$
the total number of vertices selected in this manner, and let $S = \{i(1),\dots,i(m)\}$.
It is immediate to see that
\begin{align}
\sum_{i\in S}v_{i}^2 \ge \frac{\|\bv\|_2^2}{k+1}\, .\label{eq:BoundOneK}
\end{align}
Further, letting $S^c = [p]\setminus S$,
\begin{align}
\Var_{\btheta}(\<\bv,\bX\>)&\ge \Var_{\btheta}(\<\bv,\bX\>|\bX_{S^c})
\\
&=
\sum_{i\in S}\Var_{\btheta}(v_iX_i|\bX_{S^c})  \label{eq:IdentityVar1}\\
&= \sum_{i\in S}v_i^2 \,
\frac{1}{L_0(p)}\ge \frac{\|\bv\|_2^2}{(k+1)L_0(p)}\, . \label{eq:IdentityVar2}
\end{align}
Here the  identity in Eq.~(\ref{eq:IdentityVar1}) follows because the $(X_i)_{i\in S}$ are
conditionally independent given $\bX_{S^c}$
(note that $S$ is an independent set in $G$, i.e. there is no edge
connecting two vertices in $S$), and because $(X_i)_{i\in S^c}$ constant given $\bX_{S^c}$
The expressions in Eq.~(\ref{eq:IdentityVar1}) follow from
Eqs.~(\ref{eq:VarI}) and (\ref{eq:BoundOneK}).

We therefore established also condition {\sf C2}, with $L(p) =
(k+1)L_0(p)$. This finishes the proof. 
\end{proof}

\subsection{Tractable estimation from samples}
\label{sec:EstimationFromSamples}

In this section we assume to be given $n$ i.i.d. samples
$\bX_1,\bX_2,\dots,\bX_n\sim p_{\btheta}$ and denote
by $\bX^{(n)} = (\bX_1,\bX_2,\dots,\bX_n)$ the entire dataset.
We will seek an estimator $\bhtheta=
\bhtheta(\bX^{(n)};\xi)$ that can be computed efficiently, and is
consistent in the sense that,
for a sequence $\xi_n\to 0$
\begin{align}
\bhtheta(\bX^{(n)};\xi_n)\stackrel{p}{\to }\btheta\, ,
\end{align}
in $p_{\btheta}$-probability. 

It is indeed fairly easy to construct such an estimator:
we will use an approach developed in \cite{abbeel2006learning}. Fix $i\in [p]$
and say we want to estimate $\theta_i$. Let $N_i(x_i;\bx_{\di})$ be
number of samples $\ell$ such that 
$X^{(\ell)}_i = x_i$, $\bX^{(\ell)}_\di = \bx_\di$. In formulae
\begin{align}
N_i(x_i;\bx_{\di}) =\#\Big\{\ell\in[n]:\;\; X^{(\ell)}_i = x_i,
\bX^{(\ell)}_\di = \bx_\di\Big\}\, . 
\end{align}
Then by the law of large numbers we have, almost surely and in probability
\begin{align}
\lim_{n\to\infty}\frac{N_i(1;\bzero)}{N_i(0;\bzero)} =
\frac{\prob_{\btheta}(X_i=1, \bX_{\di}=0)}{\prob_{\btheta}(X_i=0, \bX_{\di}=0)}=
\frac{\prob_{\btheta}(X_i=1|\bX_{\di}=0)}{\prob_{\btheta}(X_i=0|
  \bX_{\di}=0)}\, .
\end{align}
Now, an immediate generalization of Eq.~(\ref{eq:CondProb}) yields
\begin{align}
\frac{\prob_{\btheta}(X_i=1|\bX_{\di}=0)}{\prob_{\btheta}(X_i=0|
  \bX_{\di}=0)} = e^{2\beta k+\theta_i}\, .\label{eq:RatioProb}
\end{align}
This immediately suggests the estimator
\begin{align}
\htheta_i \equiv -2\beta k+ \log \frac{N_i(1;\bzero)}{N_i(0;\bzero)}
\, . \label{eq:EstimatorSamples}
\end{align}
This can be obviously evaluated for all vertices $i\in [p]$, in time
linear in $p$ and in $n$.
We next provide a non-asymptotic estimate on the number of samples
necessary to achieve a desired level of accuracy. 
\begin{proposition}
Let $\xi>0$ be a precision parameter and $\Delta$ an error
probability,
and let $\btheta$ be such that $\|\btheta\|_{\infty}\le
\theta_{\max}$, $\theta_{\rm max}\ge 1$. 
Then, letting $\bhtheta$ denote the estimator
(\ref{eq:EstimatorSamples}), we have
\begin{align}
\prob_{\btheta}\Big(\|\bhtheta-\btheta\|_{\infty}\le \xi\Big)\ge
1-\Delta\, ,
\end{align}
provided $n\ge e^{C_*\theta_{\rm max}} \xi^{-2}\log (p/\Delta)$, where $ C_*=
C_*(k,\beta)$ is a constant.
\end{proposition}
\begin{proof}
Fix $i\in [p]$, and let, for $x\in\{0,1\}$,
\begin{align}
q_x\equiv \prob_{\btheta}(X_i = x; \bX_{\di} = \bzero)\, .
\end{align}
Letting $D\subseteq [p]$, $|D|\le k(k-1)$, be the set of neighbors of
$\{i\}\cup\di$ in $G$, we have, by the Markov property
\begin{align}
\min_{\bx_D\in\{0,1\}^D}\prob_{\btheta}\big(X_i = x; \bX_{\di} =
\bzero\big|\bX_{D}=\bx_D\big)\le  
q_x\le \max_{\bx_D\in\{0,1\}^D}\prob_{\btheta}\big(X_i = x; \bX_{\di} = \bzero\big|\bX_{D}=\bx_D\big)\, .
\end{align}
An explicit calculation shows that there exists a constant $C =
C(k,\beta)$ such that
\begin{align}
e^{-C\theta_{\max}}\le q_x \le 1-e^{-C\theta_{\max}}\, . \label{eq:BoundProbQ}
\end{align}
Note that, for $x\in\{0,1\}$, $N_{i}(x;\bzero)$ is a binomial random
variable ${\rm Binom}(n,q_x)$. Standard tail bounds on binomial
random variables yield (for $\eps\le 1/2$)
\begin{align}
\prob_{\btheta}\Big(\big|N_i(x;\bzero)-\E N_i(x;\bzero)\big|\ge \eps
\E N_i(x;\bzero)\Big)\le
2\,\exp\Big\{-\frac{n\eps^2}{8\min(q_x,1-q_x)}\Big\}\, .
\end{align}
Using Eq.~(\ref{eq:RatioProb}) and the definition
(\ref{eq:EstimatorSamples}), together with Eq.~(\ref{eq:BoundProbQ}), we get 
\begin{align}
\prob_{\btheta}\Big(\big|e^{\htheta_i}-e^{\theta_i}\big|\ge \eps e^{\theta_i}\Big)\le
2\,\exp\Big\{-ne^{C\theta_{\rm max}}\eps^2\Big\}\, .
\end{align}
Passing to logarithms and using $|\log(1+x)|\le C|x|$ for all $x\le
1/2$, this yields, for all $\eps<1/4$, and eventually a different
constant $C$,
\begin{align}
\prob_{\btheta}\Big(\big|\htheta_i-\theta_i\big|\ge \xi \Big)\le
2\,\exp\Big\{-ne^{C\theta_{\rm max}}\xi^2\Big\}\, .
\end{align}
The proof is completed by choosing $n$ so that the right hand side is
upper bounded by $\Delta/p$ and using the union bound over the
vertices $p$.
\end{proof}

\section{Discussion and related literature}
\label{sec:Discussion}

From a mathematical point of view, the phenomenon highlighted by
Theorem \ref{thm:Main} is not new. It can be traced back to two
well-understood facts, and one simple remark:
\begin{enumerate}
\item First well-understood fact. 
The log partition function is the value of a convex optimization
  problem:
\begin{align}
\log Z(\bzero) = \max\Big\{\, F(\btau)\, :\;\;
\btau\in[0,1]^p\,\Big\}\, . \label{eq:ConvexOptimization}
\end{align}
\item Second well-understood fact. Recall that  a weak evaluation oracle for $F$ is
  an oracle that --on input $\btau,\eps$, with $\btau \in [0,1]^p$ and
  $\eps>0$-- returns $\hF$ such that $|F(\btau)-\hF|\le \eps$.
Given such an oracle, then the optimization problem
(\ref{eq:ConvexOptimization}) can be solved in polynomial time with
accuracy $C(p)\eps$, with $C(p)$ a polynomial. 

(See for instance \cite{lovasz1987algorithmic}, Lemma 2.2.4 and
Theorem 2.2.14 or \cite{grotschel1981ellipsoid}.)
\item Simple remark. We know that $F$ is differentiable with Lipschitz
  continuous gradient for $\btau\in [\delta,1-\delta]^p$, by assumption
  {\sf C2}. Further, we showed quite easily that, 
  letting  $\btau^{(0)} = (\delta,\dots,\delta)^{\sT}$, we have
  $F(\btau^{(0)})\approx F(\bzero) = \log
  h(\bzero)$. It follows that $F(\btau)$ can be --approximately--
  evaluated by `integrating' the gradient $\nabla F$ between
  $\btau^{(0)}$ and $\btau$.

Since we assume to have access to a gradient oracle (i.e. an oracle
for $\nabla F(\btau) = -\btheta_*(\btau)$), we can construct an
evaluation oracle. Hence, by the previous facts we can approximate
$\log Z(\bzero)$.
\end{enumerate}
In other words, we could have proven Theorem \ref{thm:Main} by using
the construction at point 3, and then referring to standard results in
convex optimization \cite{grotschel1981ellipsoid,lovasz1987algorithmic}. We preferred a
self-contained proof, that uses additional structure of the problem
(differentiability of $F$  and existence of an oracle for $\nabla F$).

The specific example discussed in Section \ref{sec:Example} can be regarded
as a graphical model. While we used a specific class of models, namely
antiferromagnetic Ising models  defined by Eq.~(\ref{eq:AntiFerro}),
approximating the partition function of a graphical model is --in many
cases-- intractable \cite{goldberg2003computational,sly2010computational,sly2012computational,galanis2012inapproximability,cai2012inapproximability}. Hence, the conclusion is that --generally
speaking-- \emph{it is intractable to estimate the parameters of a graphical model from
sufficient statistics} (unless the model has  special structure). 

The literature on estimating parameters and structure of a graphical
model is quite vast, and we can only provide a few pointers here. 
Traditional methods \cite{hinton1983analysing,ackley1985learning}
attempt at maximizing the likelihood function by gradient
ascent. These approaches are necessarily based on sufficient
statistics, and hence covered by our Theorem \ref{thm:Main}: in general, they
cannot be implemented in polynomial time. The bottleneck is quite
apparent in this case: evaluating the likelihood function or its
gradient requires to compute expectations with respect to
$p_{\btheta}$ which --in general-- is hard. Notice that  this
difficulty is not circumvented by regularizing the likelihood
function, as in the graphical LASSO \cite{banerjee2008model}.

A possible approach to overcome this problem was proposed in
\cite{ravikumar2010high}: the basic idea is to reconstruct the
neighborhood of a node $i$ by performing a logistic regression against
the other nodes. This approach generalizes to discrete graphical
models a method developed in \cite{MeinshausenBuhlmann} for Gaussian graphical models.
As shown in \cite{bento2009graphical}, the logistic regression method fails unless the
graphical model satisfies a `weak dependence' condition. Under the
same type of condition, the log partition function can be computed
efficiently.

Several papers
\cite{abbeel2006learning,csiszar2006consistent,bresler2008reconstruction,ray2012greedy,anandkumar2012high}
develop algorithm to learn parameters and graph structures, with
consistency guarantees under weak assumptions on the graphical
model. As stressed in Section \ref{sec:EstimationFromSamples}, these
algorithms do not make use exclusively of the sufficient statistics
and instead effectively estimate the joint distributions of subsets of
$k$ variables, with $k$ depending on the maximum degree.

On the impossibility side, Bogdanov, Mossel and Vadhan
\cite{bogdanov2008complexity} showed that estimating graphical models
with hidden nodes is intractable.
Singh and Vishnoi \cite{singh2013entropy} establish a result that is
very similar to Theorem \ref{thm:Main} with a slightly different
oracle assumption. Their proof uses the ellipsoid algorithm instead of
projected gradient, and their motivation is related to maximum entropy
rounding techniques in optimization algorithms.  Roughgarden and
Kearns  \cite{roughgarden2013marginals} establish equivalence of
various computational tasks in graphical models. Again, they use convex
duality and the ellipsoid algorithm.

\vspace{0.25cm}

Finally, I recently became aware of unpublished work by Bresler,
Gamarnik and Shah \cite{bresler2014hardness}, that also proves intractability of learning graphical
models  from sufficient statistics. Their proof strategy is broadly similar
to the one presented here but differs in the  solution of several technical obstacles.

\section*{Acknowledgements}

This work was partially supported by the NSF grant CCF-1319979 and the grants AFOSR/DARPA
FA9550-12-1-0411 and FA9550-13-1-0036. I am grateful to Guy Bresler
for discussing \cite{bresler2014hardness} prior to publication.

\appendix

\section{Simple properties of exponential families}
\label{app:Expo}

In this section we provide a self-contained proof of Proposition
\ref{prop:Expo}. {\sf Fact1} and {\sf Fact2} are immediate and we omit
their proof.

\vspace{0.25cm}

\noindent{\sf Fact3.} For any vector $\bv\in\reals^p$, $\|\bv\|_2=1$, we have
\begin{align} 
\<\bv,\Cov_{\btheta}(\bX)\bv\> = \Var_{\btheta}\Big(\<\bv,\bX\>\Big) =
\frac{1}{2}\, \E_{\btheta}\Big\{\<\bv,\bX-\bX'\>^2\Big\}\, ,
\end{align}
with $\bX$, $\bX'\sim p_{\btheta}$ independent. Letting $\bs=
\sign(\bv)$, and continuing the above chain of inequalities,
we get
\begin{align} 
\<\bv,\Cov_{\btheta}(\bX)\bv\> &\ge 2\cdot\frac{1}{2}p_{\btheta}(\bs) \,
p_{\btheta}(-\bs) \<\bv,\bs-(-\bs)\>^2 \\
&\ge 
p_{\btheta}(\bs) \,
p_{\btheta}(-\bs)\|\bv\|_1^2 \ge \min_{\bx\in
  \{0,1\}^n}p_{\btheta}(\bx)^2 \, ,
\end{align}
which yields the desired claim.

\vspace{0.25cm}

\noindent{\sf Fact4.} It is convenient to extend by continuity $\btau_*:\creals\to[0,1]^p$,
where $\creals \equiv\reals\cup\{+\infty,-\infty\}$ is the extended
real line. It is a simple analysis exercise to check that this 
extension exists and is well defined. Indeed, if
$\btheta\in\creals^{n}$
is such that $\theta_{i} = +\infty$ for all $i\in S_+\subseteq [n]$,
 $\theta_{i} = -\infty$ for all $i\in S_-\subseteq [n]$, and
 $\theta_i\in (-\infty,+\infty)$ for all $i\in S_0\equiv [n]\setminus (S_+\cup
 S_-)$, then we have 
\begin{align}
\btau_*(\btheta) &= \frac{1}{\tZ(\btheta)} \, \sum_{\substack{
\bx\in\{0,1\}^n\,  \text{s.t.}\\
\bx_{S_+} = 1, \bx_{S_-} = 0}}\!\!
\bx\; h(\bx)\; \exp\Big\{\<\btheta_{S_0},\bx_{S_0}\>\Big\}\, ,\label{eq:TauCont}\\
\tZ(\btheta) &\equiv 
\sum_{\substack{
\bx\in\{0,1\}^n\,  \text{s.t.}\\
\bx_{S_+} = 1, \bx_{S_-} = 0}}\!\!
h(\bx)\;
\exp\Big\{\<\btheta_{S_0},\bx_{S_0}\>\Big\}\, ,
\end{align}

We next prove that $\btau_*$ is injective. Indeed, assume by
contradiction that $\btau_*(\btheta_1) = \btau_*(\btheta_2)$ for
$\btheta_1\neq\btheta_2$. Then, by the intermediate value theorem,
there exists $\lambda\in [0,1]$ such that, letting $\btheta_\lambda= \lambda\, \btheta_1 +(1-\lambda)\,\btheta_2$,
\begin{align}
\btau_*(\btheta_1)-\btau_*(\btheta_2) = 
\nabla_{\btheta}\btau_*(\btheta_{\lambda}) (\btheta_1-\btheta_2)\, ,
\end{align}
and hence, letting $\bv = \btheta_1-\btheta_2$,
\begin{align}
0 = \<\bv,
\btau_*(\btheta_1)-\btau_*(\btheta_2)\> = \<\bv,\Cov_{\btheta}(\btheta_{\lambda}) \bv\>\, ,
\end{align}
which is impossible by {\sf Fact 3} above.

In order to prove that $\btau_*:\creals^p\to [0,1]^p$ is surjective,
we will proceed by induction over the problem's dimension $p$. The
claim is obvious for $p=1$. We assume next that it holds for all
dimensions up to $(p-1)$ and prove it for dimension $p$. This claim
follows by continuity, after proving that the image of $\btau_*$
contains the boundary $B_p \equiv [0,1]^p\setminus (0,1)^p$. Namely,
for each $\btau\in B_p$, there exists $\btheta\in\creals^p$ such that
$\btau_*(\btheta) =\btau$. 

To see that this is the case, fix one such $\btau$. Let $S_+ \equiv
\{i\in [n] :\, \tau_i = 1\}$,  $S_- \equiv \{i\in [n] :\, \tau_i =
0\}$, and $S_0 \equiv [p]\setminus (S_+\cup S_-)$ and note that, by assumption
$|S_0|\le n-1$. Take $\btheta$ such that $\theta_i = +\infty$ for all
$i\in S_+$ and $\theta_i = -\infty$ for all
$i\in S_-$. By Eq.~(\ref{eq:TauCont}), we have
$\btau_*(\btheta)_{S_{+}\cup S_{-}} = \btau_{S_+\cup S_-}$ and it is
therefore sufficient to check the values of $\btau_*(\btheta)$ on
$S_0$. 
Let $h_{S_+,S_-}(\bx_{S_0}) = h(\bx)$ where $\bx_{S_+} = 1_{S_+}$ and
$\bx_{S_-} = 0_{S_-}$.
Then, again by Eq.~(\ref{eq:TauCont}), we have
\begin{align}
\btau_*(\btheta)_{S_0} &= \frac{1}{\tZ(\btheta)}\sum_{\bx_{S_0}\in\{0,1\}^{S_0}}\!\!
\bx_{S_0}\; h_{S_+,S_-}(\bx_{S_0})\;
\exp\Big\{\<\btheta_{S_0},\bx_{S_0}\>\Big\}\, ,\\
\tZ(\btheta) & = \sum_{\bx_{S_0}\in\{0,1\}^{S_0}}\!\!
h_{S_+,S_-}(\bx_{S_0})\;
\exp\Big\{\<\btheta_{S_0},\bx_{S_0}\>\Big\}\, .
\end{align}
Comparison with Eq.~(\ref{eq:TauDef}) shows that --by the induction
hypothesis-- there exists $\btheta_{S_0}$ such that
$\btau_*(\btheta)_{S_0} = \btau_{S_0}$.

\vspace{0.25cm}

\noindent{\sf Fact5.} This is a standard exercise with
Legendre-Fenchel transforms and we omit its proof.

\vspace{0.25cm}

\noindent{\sf Fact6.} Recall the Gibbs variational principle \cite{MezardMontanari}
\begin{align}
A(\btheta) = \max_{q\in\cP_p} \Big\{H(q) + \E_q\big(\log
h(\bX)+\<\btheta,\bX\>\big)\Big\}\, .
\end{align}
The claim follows by comparing this expression with Eq.~(\ref{eq:InverseTransform}).

\section{Finishing the proof of Theorem \ref{thm:Main}}
\label{app:Details}

In this appendix we complete the proof of Theorem \ref{thm:Main}.
In the simplified version presented in Section \ref{sec:Main}, we 
assumed to have access to an oracle returning $\btheta_*(\btau)$ when 
queried with value $\btau$. We will show that our claim remains
correct if the oracle is replaced by a polynomial consistent
estimator $\bhtheta(\, \cdot\,, \,\cdot\,)$.
By Definition \ref{def:Estimator}, for any $\xi>0$ we have
\begin{align}
\|\bhtheta(\btau,\xi)-\btheta_*(\btau)\|_2\le \xi\, .
\end{align}

The oracle $\btheta_*(\,\cdot\,)$ is used in two points in the proof of
Theorem \ref{thm:Main}:
\begin{enumerate}
\item In the implementation of the projected gradient algorithm,
  cf. Eq.~(\ref{eq:ProjGrad}). It is queried a total of $t_0(p,\eps)
  \equiv 2p\,L(p)/\eps$ times for this purpose.
\item In calculating the approximation of $Z(\bzero)$, cf.
  Eq.~(\ref{eq:Zestimate}).
It is queried a total of $m_0(p,\eps) \equiv \lceil 4L(p)p/\eps\rceil$
times for this purpose.
\end{enumerate}
We will replace these queries with queries to
$\bhtheta(\,\cdot\,,\xi)$ with $1/\xi$ polynomial in $p$ and $1/\eps$.
Since the total number of calls is also polynomial in $p$ and $1/\xi$, this yields an
algorithm that is polynomial in $p$ and $1/\eps$. 

\vspace{0.25cm}

\noindent{\bf Queries for projected gradient.}
We denote by $\bsigma^0$, $\bsigma^1$,\dots $\bsigma^{t_0}$ the
sequence generated by the projected gradient algorithm, with
$\nabla_{\btau}F(\btau)$ approximated by $-\bhtheta(\btau,\xi)$.
Namely, we have $\bsigma^0 = \btau^0$ and, for all $t\ge 0$,
\begin{align}
\bsigma^{t+1} & =
\proj_{\delta}\Big(\bsigma^t-\frac{1}{L}\bhtheta(\bsigma^t,\xi)\Big)\, . 
\end{align}
Comparing with Eq.~(\ref{eq:ProjGrad}), we have (dropping the argument
$\xi$ of $\bhtheta$ in order to simplify the notation):
\begin{align}
\|\btau^{t+1}-\bsigma^{t+1}\|_2\le & 
\Big\|\Big(\btau^t-\frac{1}{L}\bhtheta(\btau^t)\Big)-\Big(\bsigma^t-\frac{1}{L}\bhtheta(\bsigma^t)\Big)\Big\|_2\\
\le &\Big\|\btau^t+\frac{1}{L}\nabla_{\btau}F(\btau^t)-\bsigma_t-\frac{1}{L}\nabla_{\btau}F(\bsigma^t)\Big\|_2 \\
&+\frac{1}{L}\|\bhtheta(\bsigma_t)
-\btheta_*(\bsigma_t) \|_2 +\frac{1}{L}\|\bhtheta(\btau_t)
-\btheta_*(\btau_t) \|_2 \\
 \le & \Big\|\Big(\id +\frac{1}{L}\nabla^2F(\brho^t)\Big)(\btau^t-\bsigma_t)\|_2 +\frac{2\xi}{L}\, ,
\end{align}
where the last inequality follows from the definition of $\bhtheta$,
and by applying the intermediate value theorem, with $\brho^t$ a point
on the interval $[\btau^t,\bsigma^t]$. Now, by
Eq.~(\ref{eq:HessianCov}), (\ref{eq:InverseHessian}) and assumption
{\sf C2}, we have $-L\preceq \nabla^2F(\brho^t)\preceq \bzero$, and therefore
 $\|\id +L^{-1}\nabla^2F(\brho^t)\|_2\le 1$ (in operator norm). We get therefore 
\begin{align}
\|\btau^{t+1}-\bsigma^{t+1}\|_2\le &\|\btau^{t}-\bsigma^{t}\|_2 +
\frac{2\xi}{L} \, ,
\end{align}
which, of course, implies
\begin{align}
\|\btau^{t_0}-\bsigma^{t_0}\|_2\le &\frac{2t_0\, \xi}{L}\, .
\end{align}

Notice that, for any $\btau\in D_p(\delta)$, 
\begin{align}
\|\nabla_{\btau}F(\btau)\|_2  =
\|\nabla_{\btau}F(\btau)-\nabla_{\btau}F(\btau_*)\|_2  \le L\,
\|\btau-\btau_*\|_2\le L\sqrt{p}\, .
\end{align}
Hence
\begin{align}
\big|F(\btau^{t_0})-F(\bsigma^{t_0}) \big|\le L\sqrt{p} \,
\|\btau^{t_0}-\bsigma^{t_0}\|_2\le 2t_0\,\sqrt{p} \xi\, .
\end{align}
In particular, by choosing $\xi\le \eps/(16t_0(p,\eps)\sqrt{p})$, it
follows that $|F(\btau^{t_0})-F(\bsigma^{t_0})|\le \eps/8$.

\vspace{0.25cm}

\noindent{\bf Queries for evaluating Eq.~(\ref{eq:Zestimate}).}
We let $\bsigma^{(0)}= \btau^{(0)}$, $\bsigma^{(m)} = \sigma^{t_0}$ and
$\bsigma^{(\ell)}$, $\ell\in\{1,\dots,m-1\}$ be given by interpolating
linearly. The approximation $\hZ$ in Eq.~(\ref{eq:Zestimate}) is
replaced by
\begin{align}
\log \hZ = \log h(\bzero) -
\sum_{\ell=1}^{m}\<\btheta_*(\btau^{(\ell)}),\btau^{(\ell)}-\btau^{(\ell-1)}\>\,.
\end{align}
The approximation error is bounded analogously to
Eq.~(\ref{eq:BoundSum}). We get two additional error terms
\begin{align}
\Big|\log \frac{\hZ}{Z}\Big|&\le \frac{3\eps}{4}
+|F(\btau^{t_0})-F(\bsigma^{t_0})|+\Big|\sum_{\ell=1}^m\<\bhtheta(\bsigma^{(\ell)})-\btheta_*(\bsigma^{\ell}),
\bsigma^{(\ell)}-\bsigma^{(\ell-1)}\>\Big|\\
&\le  \frac{7\eps}{8} + \|\bsigma^{(m)}-\bsigma^{(0)}\|_2\max_{\ell\in
[m]} \|\bhtheta(\bsigma^{(\ell)})-\btheta_*(\bsigma^{\ell})\|_2\\
&\le \frac{7\eps}{8} + \xi\sqrt{p}\, .
\end{align}
The latter is bounded by $\eps$ as soon as $\xi\le 1/(8\eps
\sqrt{p})$,
which is guaranteed since we already required $\xi\le
\eps/(16t_0(p,\eps)\sqrt{p})$.

\vspace{0.25cm}

We conclude that the desired approximation error is guaranteed by
using the oracle $\bhtheta$, with accuracy parameter $\xi\le
\eps/(16t_0(p,\eps)\sqrt{p})$. This concludes the proof, since it
corresponds to $(1/\xi)$  polynomial in $p$ and $(1/\eps)$.

\bibliographystyle{amsalpha}

\newcommand{\etalchar}[1]{$^{#1}$}
\providecommand{\bysame}{\leavevmode\hbox to3em{\hrulefill}\thinspace}
\providecommand{\MR}{\relax\ifhmode\unskip\space\fi MR }
\providecommand{\MRhref}[2]{%
  \href{http://www.ams.org/mathscinet-getitem?mr=#1}{#2}
}
\providecommand{\href}[2]{#2}

\end{document}